\pdfoutput=1
\documentclass[a4paper,UKenglish,cleveref, autoref, thm-restate]{lipics-v2021}

\title{Generalized Universe Hierarchies and First-Class Universe Levels}

\author{András Kovács}{Eötvös Loránd University, Hungary}{kovacsandras@inf.elte.hu}{https://orcid.org/0000-0002-6375-9781}{}

\authorrunning{A., Kovács} 

\Copyright{András Kovács} 



\relatedversion{} 

\funding{The author was supported by the European Union,
co-financed by the European Social Fund (EFOP-3.6.3-VEKOP-16-2017-00002).}

\nolinenumbers 

\hideLIPIcs  

\EventEditors{Florin Manea and Alex Simpson}
\EventNoEds{2}
\EventLongTitle{30th EACSL Annual Conference on Computer Science Logic (CSL 2022)}
\EventShortTitle{CSL 2022}
\EventAcronym{CSL}
\EventYear{2022}
\EventDate{February 14--19, 2022}
\EventLocation{G\"{o}ttingen, Germany (Virtual Conference)}
\EventLogo{}
\SeriesVolume{216}
\ArticleNo{35}

\usepackage{xcolor}
\usepackage{mathpartir}
\usepackage{todonotes}
\presetkeys{todonotes}{inline}{}
\usepackage{scalerel}
\usepackage{amssymb}
\usepackage{bm}

\theoremstyle{remark}
\newtheorem{notation}{Notation}

\theoremstyle{definition}
\newtheorem{mydefinition}{Definition}

\newtheorem{mylemma}{Lemma}

\newcommand{\Set}[1]{\mathsf{Set_{#1}}}
\newcommand{\Seti}{\mathsf{Set}}

\newcommand{\Con}{\mathsf{Con}}
\newcommand{\Ty}{\mathsf{Ty}}
\newcommand{\Tm}{\mathsf{Tm}}
\newcommand{\Sub}{\mathsf{Sub}}
\newcommand{\emptycon}{\scaleobj{.75}\bullet}
\newcommand{\U}{\mathsf{U}}
\newcommand{\El}{\mathsf{El}}
\newcommand{\id}{\mathsf{id}}
\newcommand{\ext}{\triangleright}
\newcommand{\blank}{\mathord{\hspace{1pt}\text{--}\hspace{1pt}}}
\newcommand{\mi}[1]{\mathit{#1}}
\newcommand{\p}{\mathsf{p}}
\newcommand{\q}{\mathsf{q}}
\newcommand{\Id}{\mathsf{Id}}
\newcommand{\Nat}{\mathsf{Nat}}
\newcommand{\Bool}{\mathsf{Bool}}
\newcommand{\true}{\mathsf{true}}
\newcommand{\false}{\mathsf{false}}
\newcommand{\up}{{\uparrow}}
\newcommand{\down}{{\downarrow}}
\newcommand{\Lift}{\mathsf{Lift}}
\renewcommand{\tt}{\mathsf{tt}}
\newcommand{\Acc}{\mathsf{Acc}}
\newcommand{\acc}{\mathsf{acc}}
\newcommand{\Lvl}{\mathsf{Lvl}}
\renewcommand{\U}{\mathsf{U}}
\newcommand{\Code}{\mathsf{Code}}
\newcommand{\msf}[1]{\mathsf{#1}}
\newcommand{\uir}{\msf{U^{IR}}}
\newcommand{\elir}{\msf{El^{IR}}}
\newcommand{\ult}{\U_{<}}
\newcommand{\mkMor}{\msf{mk}\!_<}
\newcommand{\unMor}{\msf{un}\!_<}
\newcommand{\mkLvl}{\msf{mk}_{\Lvl}}
\newcommand{\unLvl}{\msf{un}_{\Lvl}}

\begin{document}
\maketitle

\begin{abstract}
In type theories, universe hierarchies are commonly used to increase the
expressive power of the theory while avoiding inconsistencies arising from size
issues. There are numerous ways to specify universe hierarchies, and theories
may differ in details of cumulativity, choice of universe levels, specification
of type formers and eliminators, and available internal operations on levels. In
the current work, we aim to provide a framework which covers a large part of the
design space. First, we develop syntax and semantics for cumulative universe
hierarchies, where levels may come from any set equipped with a transitive
well-founded ordering. In the semantics, we show that induction-recursion can be
used to model transfinite hierarchies, and also support lifting operations on
type codes which strictly preserve type formers. Then, we consider a setup where
universe levels are first-class types and subject to arbitrary internal
reasoning. This generalizes the bounded polymorphism features of Coq and at the
same time the internal level computations in Agda.
\end{abstract}

\section{Introduction}
\label{sec:introduction}

Users of type theories often view universe levels as a bureaucratic detail, a
necessary annoyance in service of boosting expressive power while retaining
logical consistency. However, universe hierarchies are not going away any time
soon in practical implementations of type theory. In recent developments of
systems, we are getting more universes and more adjacent features:
\begin{itemize}
\item Agda recently added a limited cumulativity as an optional feature
  for universes \cite{agdadocs}, and the upcoming 2.6.2 version will extend the $\omega+1$
  universe hierarchy to $\omega*2$.
\item Coq added support for cumulative inductive types \cite{timany18cumulative}
  and a form of bounded universe polymorphism \cite{ziliani15unification}.
\end{itemize}
\noindent At this point, there is a veritable zoo of universe features in existing
implementations. We have perhaps even more design choices when considering the
formal metatheory of type theories. Do type formers stay in the same
universe, or take the $\sqcup$ of universes of constituent types? Can
eliminators target any universe, or do we instead use lifting operators to cross
levels? What kind of universe polymorphism do we have, can we quantify over
level bounds? Is there a type of levels, or are levels in a separate syntactic layer?

The aim of the current work is to develop semantics which covers as much as
possible from the range of sensible universe features. This way, theorists and
language implementors can grab a desired bag of features, and be able to show
consistency of their system by a straightforward translation to one of the
systems in this paper.

\paragraph*{Contributions}

\begin{enumerate}
\item In Section \ref{sec:ttgu} we describe models of type theories where
  universe levels may come from any set with a well-founded transitive ordering
  relation. We specify models as categories equipped with level-indexed diagrams
  of families, as a variation on categories with families. Each morphism of
  levels is mapped to a lifting operation on terms and types. By varying the
  preservation properties of lifting operations, we can describe a range of
  stratification features, from two-level type theory to cumulative universes.
\item In Section \ref{sec:semantics} we use induction-recursion to model the mentioned theories. We
  model the strongest formulations for lifting and universes, namely cumulative universes with
  Russell-style type decoding.
\item In Section \ref{sec:ttfl} we describe type theories with internal types
  for levels and level morphisms, and extend the previous inductive-recursive
  semantics to cover these as well. Here, we can additionally represent various
  universe polymorphism features and level computations.
\end{enumerate}

We provide an Agda formalization of the contents of the paper at
\url{https://github.com/AndrasKovacs/universes/tree/master/agda}. The
formalization is not complete, as we skip proofs involving an excessive number
of equality coercions (which are more suited to informal reasoning, using
equality reflection), and instead focus on the key points.

\section{Metatheory}
\label{sec:metatheory}

We work in a Martin-Löf type theory which has the following features.
\begin{itemize}
  \item Two universes named $\Set0$ and $\Set1$, where $\Set0$ supports inductive-recursive types
    (IR) as specified by Dybjer and Setzer \cite{dybjer99finite}. We may omit the universe indices
    if they can be inferred or if we work over arbitrary indices.
  \item Function extensionality and uniqueness of identity proofs
    (UIP). Additionally, we assume equality reflection in this paper, thus
    working in extensional type theory, to avoid noise from equality transports.
  \item We write function types as $(x : A)\to B$ with $\lambda\,x.\,t$ inhabitants. We may group
    multiple arguments with the same type, as in $(x\,y : A) \to B$.  We have $\Sigma$-types as $(x :
    A) \times B$, with pairing as $(t,\, u)$.  We have $\top$ as the unit type with inhabitant
    $\tt$, $\bot$ as the empty type, and $\Bool$ with $\true$ and $\false$
    inhabitants. Propositional identity is written as $t = u$ (coinciding with definitional
    equality).
  \item We occasionally use $\{x : A\} \to B$ for an Agda-like notation for
    function types with implicit arguments. We usually omit implicit
    applications but may explicitly write them as $t\,\{u\}$.  We may omit
    implicit function types altogether if it is clear where certain variables
    are quantified.
\end{itemize}

\section{Generalized Universe Hierarchies}
\label{sec:ttgu}

In this section, we first describe notions of models for type theories with
generalized universes, and discuss several variations of universes and lifting
operations. Then, we pick a concrete variant (the strongest, in a sense)
and construct a model for it in the metatheory.

For the basic structure of typing contexts and substitutions, let us review
categories with families.

\subsection{Categories with Families}
\label{sec:categories_with_families}

\begin{mydefinition}
A \emph{category with family} (cwf) \cite{Dybjer96internaltype} consists of the following data:
\begin{itemize}
\item A category with a terminal object. We denote the set of objects as $\Con :
  \Seti$ and use capital Greek letters starting from $\Gamma$ to refer to
  objects. The set of morphisms is $\Sub : \Con \to \Con \to \Seti$, and we use
  $\sigma$, $\delta$ and so on to refer to morphisms. The terminal object is
  $\emptycon$ with unique morphism $\epsilon : \Sub\,\Gamma\,\emptycon$. In
  initial models (that is, syntaxes) of type theories, objects correspond to
  typing contexts, morphisms to parallel substitutions and the terminal object to
  the empty context; this informs the naming scheme.
\item A \emph{family structure}, containing $\Ty : \Con \to \Seti$ and $\Tm :
  (\Gamma : \Con) \to \Ty\,\Gamma \to \Seti$, where $\Ty$ is a presheaf over the
  category of contexts and $\Tm$ is a presheaf over the category of elements of
  $\Ty$. This means that both types ($\Ty$) and terms ($\Tm$) can be
  substituted, and substitution has functorial action. We use $A$, $B$, $C$ to
  refer to types and $t$, $u$, $v$ to refer to terms, and use $A[\sigma]$ and
  $t[\sigma]$ for substituting types and terms. Additionally, a family structure
  has \emph{context comprehension} which consists of a context extension
  operation $\blank\ext\blank : (\Gamma : \Con) \to \Ty\,\Gamma \to \Con$
  together with an isomorphism $\Sub\,\Gamma\,(\Delta\ext A) \simeq ((\sigma :
  \Sub\,\Gamma\,\Delta) \times \Tm\,\Gamma\,(A[\sigma]))$ which is natural in
  $\Gamma$.
\end{itemize}
\end{mydefinition}

\noindent From the comprehension structure, we recover the following notions:

\begin{itemize}
\item By going right-to-left along the isomorphism, we recover \emph{substitution extension}
      $\blank,\blank : (\sigma : \Sub\,\Gamma\,\Delta) \to \Tm\,\Gamma\,(A[\sigma]) \to \Sub\,\Gamma\,(\Delta\ext A)$. This means
      that starting from $\epsilon$ or the identity substitution $\id$, we can iterate $\blank,\blank$
      to build substitutions as lists of terms.
\item By going left-to-right, and starting from $\id : \Sub\,(\Gamma\ext A)\,(\Gamma\ext A)$, we recover
      the \emph{weakening substitution} $\p : \Sub\,(\Gamma\ext A)\,\Gamma$ and the \emph{zero variable}
      $\q : \Tm\,(\Gamma\ext A)\,(A[\p])$.
\item By weakening $\q$, we recover a notion of variables as De Bruijn indices. In general, the $n$-th
      De Bruijn index is defined as $\q[\p^{n}]$, where $\p^{n}$ denotes $n$-fold composition.
\end{itemize}

There are other ways for presenting the basic categorical structure of models,
which are nonetheless equivalent to cwfs, including natural models
\cite{awodey18natural} and categories with attributes \cite{cartmellthesis}. We
use the cwf presentation for its immediately algebraic character and closeness
to conventional explicit substitutions. We consider the syntax of a type theory
to be its initial model.

\begin{notation}As De Bruijn indices are hard to read, we will mostly use
nameful notation for binders. For example, assuming $\Nat : \{\Gamma : \Con\}
\to \Ty\,\Gamma$ and $\Id : \{\Gamma : \Con\}(A : \Ty\,\Gamma) \to
\Tm\,\Gamma\,A \to \Tm\,\Gamma\,A \to \Ty\,\Gamma$, we may write $\emptycon \ext
(n : \Nat) \ext (p : \Id\,\Nat\,n\,n)$ for a typing context, instead of using
numbered variables or cwf combinators as in $\emptycon \ext \Nat \ext
\Id\,\Nat\,\q\,\q$.
\end{notation}

\begin{notation}
In the following, we will denote families by ($\Ty$,$\Tm$) pairs and overload context
extension $\blank\ext\blank$ for different families.
\end{notation}

A family structure may be closed under certain \emph{type formers}. For example,
we may close a family over function types by assuming $\Pi : (A : \Ty\,\Gamma)
\to \Ty\,(\Gamma\ext A) \to \Ty\,\Gamma$ together with abstraction, application,
$\beta\eta$-rules, and equations for the action of substitution on type and term
formers.

In the following, whenever we introduce a type or term former, we always assume
that it is natural with respect to substitution, i.e.\ all type and term formers
have a corresponding substitution rule. This convention could be made precise by
working in a framework for higher-order abstract syntax, where all specified
structure is automatically stable under substitution
\cite{cubicalnorm,uemura,bocquet2021relative}. While this can be effective at
reducing formal clutter, this paper only presents models which are technically
straightforward, so we choose not to use higher-order signatures, in order to
make the presentation more direct.

\subsection{Morphisms and Inclusions of Families}
\label{sec:morphisms}

In the rest of the paper we make use of categories equipped with possibly
multiple family structures, which serves as basis for specifying universe
hierarchies. However, it is not very useful to simply have multiple copies of
family structures together with their type formers. In that case, every
constructor and eliminator of every type former stays in the same family, and
there is no interaction between families, and the most we can do is to mix them
together in typing contexts. In this subsection we describe several ways of
crossing between families.

\begin{mydefinition}
A \emph{family morphism} $F$ between ($\Ty_0$, $\Tm_0$) and ($\Ty_1$, $\Tm_1$)
families consists of natural transformations mapping types to types and terms to
terms, which preserves context extensions up to context isomorphism, i.e.\ we
have that $(\Gamma \ext F\,A) \simeq (\Gamma \ext A)$, where $\simeq$ denotes
existence of an invertible context morphism.
\end{mydefinition}

Family morphisms are restrictions of so-called \emph{weak morphisms}
\cite{dependentrightadjoints} (or \emph{pseudomorphisms}
\cite{kaposi2019gluing}) of cwfs: a weak morphism which has the identity action
on the base category is exactly a family morphism.

\begin{mylemma} Every family morphism has invertible action on terms, i.e.\ there
  is an $F^{-1} : \Tm\,\Gamma\,(F\,A) \to \Tm\,\Gamma\,A$.
\end{mylemma}

\begin{proof}
From the $\ext$-preservation isomorphism and the defining isomorphisms of
comprehension, we get $\q' : \Tm\,(\Gamma \ext F\,A)\,(A[\p])$ such that $F\,\q'
= \q$ and $\q'[\p,\,F\,\q] = \q$. Now, for $t : \Tm\,\Gamma\,(F\,A)$, we
define $F^{-1}\,t$ as $\q'[\id,\,t] : \Tm\,\Gamma\,A$. We get the following:
\begin{alignat*}{3}
  & F(F^{-1}\,t) = F(\q'[\id,\,t]) = (F\,\q')[\id,\,t] = \q[\id,\,t] = t \\
  & F^{-1}(F\,t) = \q'[\id,\,F\,t] = \q'[\id,\,(F\,\q)[\id,\,\,t]] = \q'[\p,\,F\,\q][\id,\,t]
    = \q[\id,\,t] = t
\end{alignat*}
More concisely, $F$ is invertible on the generic term $\q$, which implies invertibility
on any term.
\end{proof}

\begin{notation}
In the following, we will write $\Lift : \Ty_0\,\Gamma \to \Ty_1\,\Gamma$ for
the action of some morphism on types, $\up : \Tm_0\,\Gamma\,A \to
\Tm_1\,\Gamma\,(\Lift\,A)$ for the action on terms, and $\down$ for the inverse
action on terms. We will also call the action on types \emph{type lifting}
and the action on terms \emph{term lifting}.
\end{notation}

We may think about the relation between \emph{modalities} and morphisms. The
main difference is that morphisms impose no structural restrictions on variables
and contexts. More concretely, every $\Lift$ is dependent right adjoint
\cite{dependentrightadjoints} to the identity functor on the base category, as
we have $\Tm\,(\msf{Id}\,\Gamma)\,A \simeq \Tm\,\Gamma\,(\Lift\,A)$. Hence,
every morphism can be viewed as a degenerate modality.

Assume family structures ($\Ty_0$, $\Tm_0$) and ($\Ty_1$, $\Tm_1$) and a
morphism between them. This corresponds to a basic version of \emph{two-level
type theory} \cite{twolevel}. This theory has an interpretation in presheaves
over the category of contexts of some chosen model of a type theory, where
($\Ty_0$, $\Tm_0$) is modeled using structure in the chosen model, and ($\Ty_1$,
$\Tm_1$) is modeled using presheaf constructions. More illustratively, this
means interpreting ($\Ty_1$, $\Tm_1$) as a metaprogramming layer which can
generate object-level constructions in the ($\Ty_0$, $\Tm_0$) layer. Lifted
types correspond to types of object-level terms; for example, $\Bool_0 :
\Ty_0\,\Gamma$ is the object-level type of Booleans, while $\Lift\,\Bool_0$ is
the meta-level type of $\Bool_0$-terms, and $\Bool_1 : \Ty_1\,\Gamma$ is the
type of meta-level Booleans. It is possible to compute a $\Bool_0$ from a
$\Bool_1$. Given $b : \Tm_1\,\Gamma\,\Bool_1$, we can construct
$\down(\mathsf{if}\,b\,\mathsf{then}\,\up\true_0\,\mathsf{else}\,\up\false_0) :
\Tm_0\,\Gamma\,\Bool_0$. But there is no way to compute a $\Bool_1$ from a
$\Bool_0$: we can try to lift the input, but there is no elimination rule for
$\Lift\,\Bool_0$ in $\Ty_1$.

Hence, plain family morphisms can model a metaprogramming hierarchy, but
currently we are aiming for ``sizing'' hierarchies instead. This means that we
want to eliminate from any family to any other family which is connected by a
morphism.

\begin{mydefinition}\label{def:inclusion}
A \emph{family inclusion} is a family morphism which preserves all type and term
formers. This assumes that every type former which is contained in the source
family, is also contained in the target family.
\end{mydefinition}

\noindent Some examples for preservation equations for type and term formers:
\begin{alignat*}{3}
  & \Lift\,(\Pi\,(x : A) B)   && =\,\,\,\,&& \Pi\,(x : \Lift\,A)(\Lift\, (B[x \mapsto\,\down x]))\\
  & \up(\lambda\,(x : A).\,t) && =&& \lambda\,(x : \Lift\,A).\,\up(t[x\mapsto\down x])\\
  & \Lift\,\Bool_0            && =&& \Bool_1\\
  & \up\true_0                && =&& \true_1
\end{alignat*}
In general, we can skip specifying preservation for $\down$, since it follows
from $\up$ preservation equations.

Assume an inclusion from ($\Ty_0$, $\Tm_0$) to ($\Ty_1$, $\Tm_1$). Now, we can
eliminate from $\Bool_0$ to $\Bool_1$. If we have some $b :
\Tm_0\,\Gamma\,\Bool_0$, we also have $\up b :
\Tm_1\,\Gamma\,(\Lift\,\Bool_0)$, hence $\up b :
\Tm_1\,\Gamma\,\Bool_1$. Then, we can use $\Bool_1$ elimination, as in
$\mathsf{if}\,\up b\,\,\mathsf{then}\,\true_1\,\mathsf{else}\,\false_1 :
\Tm_1\,\Gamma\,\Bool_1$. The $\up$ computation ensures that the eliminator
computes appropriately on canonical terms: if $b$ is $\true_0$, we get
$\up \true_0 = \true_1$ as the if-then-else scrutinee.

A family inclusion corresponds to a \emph{cumulative hierarchy} consisting of
two families: every type former of the smaller family is included in the larger
family, with the same elimination rules.

\begin{mydefinition}\label{def:strict_inclusion}
A \emph{strict family inclusion} between ($\Ty_0$, $\Tm_0$) and ($\Ty_1$,
$\Tm_1$) is a family inclusion ($\Lift$, $\up$, $\down$) for which the following equations hold:
\begin{alignat}{2}
  & (\Gamma \ext \Lift\,A) &&= (\Gamma \ext A)     \label{eq:cumcon}    \\
  & \Tm_1\,\Gamma\,(\Lift\,A) &&= \Tm_0\,\Gamma\,A  \label{eq:cumtm}     \\
  & \up t &&= t                                    \label{eq:cumlift}
\end{alignat}
\end{mydefinition}

A strict inclusion corresponds to Sterling's \emph{algebraic cumulativity}
\cite{sterling2019algebraic}. The additional equations are a matter of
convenience: they allow us to omit term liftings in informal syntax\footnote{In
  a proof assistant, often we would still have to explicitly transport along the
  strict inclusion equations.}. Most of the time we can also omit level
annotations on term formers. For example, we have $\true_0 :
\Tm_0\,\Gamma\,\Bool_0$, but also $\true_0 : \Tm_0\,\Gamma\,(\Lift\,\Bool_0)$,
hence $\true_0 : \Tm_0\,\Gamma\,\Bool_1$. Moreover, $\true_0$ is definitionally
equal to $\true_1$, since $\true_0 =\,\up \true_0 = \true_1$. Thus, using
simply $\true$ is fine whenever the family is clear from context.

The definitional equality of $\true_0$ and $\true_1$ is important; without it
canonicity would fail, since $\true_0$, $\false_0$, $\true_1$ and $\false_1$
would be four definitionally distinct inhabitants of $\Bool_1$. See Luo
\cite{luo2012notes} for a discussion of related issues with cumulativity. It is
not sufficient to specify a strict inclusion just by equations \ref{eq:cumcon}
and \ref{eq:cumtm} in Definition \ref{def:strict_inclusion}. We need $\up$
together with equation \ref{eq:cumlift} to identify term formers in different
families. The other direction $\down t = t$ is immediately derivable.

\subsection{Level Structures}
\label{sec:level_structures}

We would like to describe a range of setups with multiple families and morphisms
between them. In this subsection we describe the indexing structures for such
family diagrams. First, we specify a notion of well-foundedness, which will be
used to preclude size paradoxes in universe hierarchies.

\begin{mydefinition}
The \emph{accessibility predicate} on relations is defined by the following
inductive rules:
\begin{alignat*}{3}
  & \Acc : \{A : \Seti\}\to(R : A \to A \to \Seti) \to A \to \Seti \\
  & \acc : \{a : A\} \to ((\mi{a'} : A) \to R\,\mi{a'}\,a \to \Acc\,R\,\mi{a'}) \to \Acc\,R\,a
\end{alignat*}
\end{mydefinition}
\noindent See \cite{aczel1977introduction} and \cite[Section~10.3]{hottbook} for
further exposition. An inhabitant of $\Acc\,R\,a$ proves that starting from $a :
A$, all descending $R$-chains must be finite. This is ensured by the universal
property of the inductive definition.

\begin{mylemma}\label{lem:accprop}
All inhabitants of $\Acc\,R\,a$ are equal \cite[Lemma 10.3.4]{hottbook}. In
other words, accessibility is proof-irrelevant.
\end{mylemma}

\begin{mydefinition}
A relation $R : A \to A \to \Seti$ is \emph{well-founded} if $(a : A) \to
\Acc\,R\,a$.
\end{mydefinition}

\begin{mydefinition} A \emph{level structure} consists of the following components:
\begin{alignat*}{3}
  & \Lvl                  &&: \Set0 \\
  & \blank\!<\!\blank     &&: \Lvl \to \Lvl \to \Set0 \\
  & \mathsf{<\!prop }     &&: (p\,q : i < j) \to p = q \\
  & \blank\!\circ\!\blank &&: j < k \to i < j \to i < k \\
  & \mathsf{<\!wf}        &&: (i : \Lvl) \to \Acc\,<\,i
\end{alignat*}
\end{mydefinition}

\noindent We overload $\Lvl$ to refer to a given level structure and also its
underlying set. In short, a level structure is a set together with a transitive
well-founded relation.

\begin{mydefinition}
A \emph{family diagram} over $\Lvl$ maps each $i : \Lvl$ to a family structure
($\Ty_i$, $\Tm_i$), and each $p : i < j$ to a family inclusion
($\Lift_{i}^{j}\,p$, $\up_{i}^{j}p$, $\down_{i}^{j}p$) between ($\Ty_i$,
$\Tm_i$) and ($\Ty_j$, $\Tm_j$). Moreover, the mapping is functorial, so
$\Lift_{i}^{k}\,(p\circ q)\,A = \Lift_{j}^{k}\,p\,(\Lift_{i}^{j}\,q\,A)$, and
similarly for $\up_{i}^{j}p$ and $\down_{i}^{j}p$. A \emph{strict family diagram}
is a family diagram where each inclusion is strict.
\end{mydefinition}

\begin{notation}
Sometimes we omit some of the $i$, $j$, $p$ annotations from type and term
liftings, if they are clear from context.
\end{notation}

Our choice of level structures and diagrams is motivated by the following.  First, we do not need
identity morphisms in levels, because they would be mapped to trivial liftings, which are not
interesting in our setting. Second, we do not need proof-relevant level morphisms, since any
parallel pair of morphisms gives rise to isomorphic types. Concretely, given $p : i < j$ and $q : i
< j$ such that $p \neq q$, we have $\Tm_j\,\Gamma\,(\Lift\,p\,A) \simeq \Tm_i\,\Gamma\,A \simeq
\Tm_j\,\Gamma\,(\Lift\,q\,A)$, and since $\Lift\,p\,A$ and $\Lift\,q\,A$ are in the same family, we
can internally prove them isomorphic using function types and identity types. That said, every
construction in this paper would still work with direct categories as level structures.

\subsection{Universes}
\label{sec:universes}

At this point, we can talk about family diagrams, but no previously seen type
former depends on levels in an interesting way. For example, $\Bool_i$ has the
same inhabitants as $\Bool_j$, for any $i$ and $j$. Universes introduce
dependency on levels, by serving as classifiers for smaller families internally
to larger families.

\begin{mydefinition} A family diagram supports \emph{universe formation} if it supports the following:
\begin{alignat*}{3}
  & \U             &&: (i\,j : \Lvl) \to i < j \to \Ty_j\,\Gamma\\
  & \mathsf{LiftU} &&: \Lift_{j}^{k}\,p\,(\U\,i\,j\,q) = \U\,i\,k\,(p \circ q)
\end{alignat*}
\end{mydefinition}
We also need a way to pin down universes as classifiers. We consider two variants.
\begin{mydefinition}
A family diagram has \emph{Coquand universes} \cite{coquandnorm} if it has universe formation and
additionally supports $\El : \Tm_j\,\Gamma\,(\U\,i\,j\,p) \to \Ty_i\,\Gamma$, and
its inverse $\Code : \Ty_i\,\Gamma \to \Tm_j\,\Gamma\,(\U\,i\,j\,p)$.
\end{mydefinition}
\begin{mydefinition}
A family diagram has \emph{Russell universes} if it has Coquand universes and
additionally satisfies $\Tm_j\,\Gamma\,(\U\,i\,j\,p) = \Ty_i\,\Gamma$ and
$\El\,t = t$.
\end{mydefinition}

The move from Coquand to Russell universes is fairly similar to the move from
inclusions to strict inclusions. The Russell variant makes it possible to
informally omit $\El$ and $\Code$. Likewise, the $\El\,t = t$ condition ensures
appropriate naturality. If we only assumed $\Tm_j\,\Gamma\,(\U\,i\,j\,p) =
\Ty_i\,\Gamma$ but not Coquand universes, we would not be able to prove that a
$t : \Tm_j\,\Gamma\,(\U\,i\,j\,p)$ substituted \emph{as a term} is the same
thing as $t$ substituted \emph{as a type}. Both would be written as $t[\sigma]$
in our notation, but they involve different $\blank[\blank]$ operations.

Unlike every other type or term former, there is no lifting computation rule for
$\El$ and $\Code$. Intuitively, the issue is that we would need to relate type
lifting and term lifting, but while term lifting is invertible, type lifting is
not. $\Lift$ sends a $\Ty_i\,\Gamma$ to a $\Ty_j\,\Gamma$, and $\Ty_j\,\Gamma$
is not isomorphic to $\Ty_i\,\Gamma$, because it contains more universes. So,
for example, lifting $\Bool_0 : \Ty_0\,\Gamma$ as a type to $\Ty_1\,\Gamma$
yields $\Bool_1$, but lifting $\Bool_0$ as a term yields $\Bool_0$.

Assuming Coquand or Russell universes and $p : i < j$, we can recover
polymorphic functions, for example, we may have $\mi{id} : \Pi(A : \U\,i\,j\,p)
(\Lift\,p\,(\El\,A) \to \Lift\,p\,(\El\,A))$ for the polymorphic identity
function. Here, we quantify over terms of $\U$, and since every type former
stays on the same level (including $\Pi$), we have to $\Lift$ the types in the
codomain to match the level of the domain. We can also recover large
elimination, for example as in
\[
(\lambda\,(b : \Bool_j).\,\mathsf{if}\,b\,\mathsf{then}\,\Code\,\top_i\,\mathsf{else}\,\Code\,\bot_i)
: \Tm_j\,\Gamma\,(\Bool_j \to \U\,i\,j\,p).
\]

\section{Semantics}
\label{sec:semantics}

In this section we give a model for a type theory with generalized
universes. Let us make the notion of model concrete first.

\begin{mydefinition}[Notion of model for a type theory with generalized universes (TTGU)]
Fix a $\Lvl$ structure. A model for TTGU consists of
\begin{enumerate}
  \item A base category ($\Con$, $\Sub$) with a terminal object $\emptycon$.
  \item A strict family diagram ($\Ty_i$, $\Tm_i$) over $\Lvl$, supporting Russell
        universes, and each family structure is closed under the same basic type formers.
\end{enumerate}
\end{mydefinition}

The choice of available basic type formers is up to personal taste, and it will
not significantly affect the following model construction.

Both in families and universes we choose the stricter formulation, since if we
give a model which proves the strict syntax consistent, we immediately get a model which
proves the weak syntax consistent\footnote{We always get \emph{initial} and \emph{terminal}
models automatically, because of the algebraic character of the theories in this paper. We
also get a \emph{freely generated} strict model from a weak model, from the left adjoint
of the functor which forgets the strictness equations. But none of these tricks can be used
to automatically get a consistency proof.}.

\subsection{Inductive-Recursive Codes}
\label{sec:inductive_recursive_codes}

The task is to interpret the $\Lvl$-many universes of TTGU using an assumed
metatheoretic feature. For this, we need to define a $\Lvl$-indexed type of type
codes. Since $\Lvl$ and $\!\blank<\!\blank$ can be arbitrary, we effectively
need to define transfinite hierarchies of codes. We use an inductive-recursive
\cite{dybjer99finite} definition for the following reasons.

First, induction-recursion is already supported in the Agda proof assistant, and
it is very useful to be able to sketch out ideas in a machine-checked
setting. It would be much harder to do the same when developing semantics in set
theory.

Second, could we use type-theoretic features with simpler specifications than
induction-recursion, such as super universes \cite{Palmgren98onuniverses} or
Mahlo universes \cite{setzer00mahlo}? These are sufficient to model transfinite
hierarchies. However, using these it is not clear how to additionally support
the strict type former preservation property of $\Lift$\footnote{Palmgren calls
this property as having \emph{recursive sub-universes}
\cite{Palmgren98onuniverses}.}.

Therefore, we give a custom definition using induction-recursion, which
corresponds more directly to TTGU structure. Our definition is essentially the
same as McBride's redundancy-free hierarchy in
\cite[Section~6.3.1]{mcbride2015datatypes}, but we generalize levels from
natural numbers to arbitrary level structures.

\begin{mydefinition}[Codes for the universe]
Assume $i : \Lvl$ and $f : (j : \Lvl) \to j < i \to \Set0$. We define $\uir$ and
$\elir$ by induction-recursion:
\begin{alignat*}{5}
  & \uir   &&: \Set0
     && \rlap{$\elir : \uir \to \Set0$} \\
  & \U'    &&: (j : \Lvl) \to j < i \to \uir
     && \elir\,(\U'\,j\,p) &&= f\,j\,p \\
  & \Pi'   &&: (A : \uir) \to (\elir\,A \to \uir) \to \uir \hspace{2em}
     && \elir\,(\Pi'\,A\,B) &&= (a : \elir\,A) \to \elir\,(B\,a)\\
  & \bot' &&: \uir
     && \elir\,\bot'      &&= \bot\\
  & \Bool' &&: \uir
     && \elir\,\Bool'     &&= \Bool
\end{alignat*}
\end{mydefinition}
We use the prime accents ($'$) to disambiguate inductive-recursive codes from
type formers in TTGU or the metatheory. For basic type formers, we only include
codes for function types, the empty type, and $\Bool$. Other type formers are
straightforward to add (and we do have more in the Agda formalization).

\begin{notation}
We may write $\uir_{i\,f}$ and $\elir_{i\,f}$ in order to make parameters explicit.
\end{notation}

($\uir$, $\elir$) can be viewed as a \emph{universe operator}: given semantics
for an initial segment of $\Lvl$ (given by $i$ and $f$), we create a new
universe which is closed under basic type formers, and also closed under all
sets in $f$ by the way of $\U'$. Most importantly, this operation can be
transfinitely iterated. We first define universes for initial segments of
$\Lvl$, by induction on the accessibility of levels:
\begin{alignat*}{3}
  & \ult\,: (i : \Lvl)\{p : \Acc\,(\blank\!<\!\blank)\,i\} \to (j : \Lvl) \to j < i \to \Set0 \\
  & \ult\,i\,\{\acc\,f\}\,j\,p = \uir_{j\,(\ult\,j\,\{f\,j\,p\})}
\end{alignat*}

\begin{mydefinition}[Semantic universe]
Since every level is accessible, we can define the full semantic hierarchy and
its decoding function.
\begin{alignat*}{3}
  &\U : \Lvl \to \Set0
     && \El : \{i : \Lvl\} \to \U\,i \to \Set0 \\
  &\U\,i = \uir_{i\,(\ult\,i\,\{\mathsf{<\!wf}\,i\})}\hspace{1em}
     && \El\,\{i\} = \elir_{i\,(\ult\,i\,\{\mathsf{<\!wf}\,i\})}
\end{alignat*}
\end{mydefinition}

\begin{mylemma}\label{lem:ucomp}
Assuming $p : i < j$, we have the computation rule $\ult\,j\,p = \U\,j$. Proof: we
may assume that any witness for $\Acc\,(\blank\!<\!\blank)\,i$ is of the form $\acc\,f$ for
some $f$. Then the equation becomes $\uir_{j\,(\ult\,j\,\{f\,j\,p\})} =
\uir_{j\,(\ult\,j\,\{\mathsf{<\!wf}\,j\})}$, but by Lemma \ref{lem:accprop}
the $f\,j\,p$ and $\mathsf{<\!wf}\,j$ witnesses are equal. \qed
\end{mylemma}

\begin{mydefinition}[Semantic $\Lift$]
We define by induction on $\uir$ a function with type $(p : i < j)\to(A : \U\,i)
\to (A' : \U\,j)\times(\El\,A' = \El\,A)$. However, for the sake of clarity, we
present this here as two (mutual) functions:
\begin{alignat*}{3}
  &\Lift\,      &&: (p : i < j) \to \U\,i \to \U\,j\\
  &\msf{ElLift} &&: (p : i < j)\to(A : \U\,i) \to \El\,(\Lift\,A) = \El\,A
\end{alignat*}
Let us look at $\Lift$ first:
\begin{alignat*}{3}
  &\Lift\,p\,(\U'\,k\,q)    &&= \U'\,k\,(p \circ q)\\
  &\Lift\,p\,(\Pi'\,A\,B)   &&= \Pi'\,(\Lift\,p\,A)\,(\lambda\,a.\,\Lift\,p\,(B\,a))\\
  &\Lift\,p\,\bot'          &&= \bot'\\
  &\Lift\,p\,\Bool'         &&= \Bool'
\end{alignat*}
Above, the $\Pi'$ definition is well-typed by $\msf{ElLift}\,p\,A$. For the proof of $\msf{ElLift}$,
the only interesting case is $\U'$. Here, we need to show $\ult\,j\,k\,(p \circ q) = \ult\,i\,k\,q$,
but by Lemma \ref{lem:ucomp} both sides are $\U\,k$.
\end{mydefinition}

\begin{mylemma}\label{lem:liftprop} Properties of $\Lift$:
\begin{enumerate}
  \item $\Lift$ preserves all basic type formers; this is immediate from the definition.
  \item $\Lift$ is functorial, i.e.\ $\Lift\,(p \circ q)\,A = \Lift\,p\,(\Lift\,q\,A)$. This follows
    by induction on $A$, and we make use of the irrelevance of $\blank\!<\!\blank$ in the $\U'$ case.
    \qed
\end{enumerate}
\end{mylemma}

\subsection{Inductive-Recursive Model of TTGU}

We give a model of TTGU in this section.

\begin{notation} To avoid name clashing between components of the model and metatheoretic
definitions, we use \textbf{bold} font to refer to TTGU components.
\end{notation}

\begin{mydefinition}[Base category]
The base category is simply the category of sets and functions in $\Set0$, i.e.\ $\bm{\Con} =
\Set0$, $\bm{\Sub}\,\Gamma\,\Delta = \Gamma \to \Delta$, and the terminal object is $\top$.
\end{mydefinition}

\begin{mydefinition}[Family diagram]
We map $i : \Lvl$ to a family structure as follows.
\begin{alignat*}{5}
  &\bm{\Ty}_i\,\Gamma = \Gamma \to \U\,i \hspace{2em} \bm{\Tm}_i\,\Gamma\,A = (\gamma : \Gamma) \to \El\,(A\,\gamma)
\end{alignat*}
Type and term substitution are given by composition with some function $\sigma : \Gamma
\to \Delta$. Comprehension structure is given by $\Gamma \bm{\ext} A =
(\gamma : \Gamma) \times \El\,(A\,\gamma)$. Type lifting along $p : i < j$ is
as follows:
\begin{alignat*}{5}
  & \bm{\Lift}_{i}^{j}p : \bm{\Ty}_i\,\Gamma \to \bm{\Ty}_j\,\Gamma\\
  & \bm{\Lift}_{i}^{j}p\,A = \lambda\,\gamma.\,\Lift_{i}^{j}p\,(A\,\gamma)
\end{alignat*}
Now, two of the strict inclusion equations follow from $\msf{ElLift}$, namely
$(\Gamma \bm{\ext} \bm{\Lift}_{i}^{j}p\,A) = (\Gamma \bm{\ext} A)$ and
$\bm{\Tm}_j\,\Gamma\,(\bm{\Lift}_{i}^{j}p\,A) = \bm{\Tm}_i\,\Gamma\,A$. Thus, we
can just define term lifting as $\bm{\up}_{i}^{j}\!p\,t = t$ and
$\bm{\down}_{i}^{j}\!p\,t = t$. Basic type formers are as follows.
\begin{alignat*}{5}
  & \bm{\Pi}\,A\,B = \lambda\,\gamma.\,\Pi'\,(A\,\gamma)\,(\lambda\,\alpha.\,B\,(\gamma,\,\alpha))
  & \hspace{2em}\bm{\bot}_i = \lambda\,\gamma.\,\bot'
  & \hspace{2em}\bm{\Bool}_i = \lambda\,\gamma.\,\Bool'
\end{alignat*}
$\bm{\Lift}_{i}^{j}p$ preserves type formers by Lemma \ref{lem:liftprop}. We
define basic term formers and eliminators using metatheoretic features, e.g.\ $\bm{\true}_i
= \lambda\,\gamma.\,\true$ and $(\bm{\lambda}_i\,x.\,t) =
\lambda\,\gamma\,\alpha.\,t\,(\gamma,\,\alpha)$. Note that since semantic term
formers are just external constructors, they do not depend on levels, so
e.g.\ $\bm{\true}_i$ is the same at all $i$. This implies that
$\bm{\up}_{i}^{j}\!p$ preserves term formers as well, so ($\bm{\Lift}_{i}^{j}p$,
$\bm{\up}_{i}^{j}\!p$, $\bm{\down}_{i}^{j}\!p$) is a strict family inclusion.

We define universes as $\bm{\U}\,i\,j\,p = \lambda\,\gamma.\,\U'_i\,j\,p$.  With
this, $\bm{\Lift}_{j}^{k}\,p\,(\bm{\U}\,i\,j\,q) = \bm{\U}\,i\,k\,(p \circ q)$
follows by the definition of semantic $\Lift$. The Russell universe equation
$\bm{\Tm}_j\,\Gamma\,(\bm{\U}\,i\,j\,p) = \bm{\Ty}_i\,\Gamma$ follows from Lemma
\ref{lem:ucomp}, so we can define $\bm{\El}$ and $\bm{\Code}$ as identity functions.

\end{mydefinition}

\begin{theorem}[Consistency of TTGU]
There is no closed syntactic term of $\bot_i$ for any $i$.
\end{theorem}
\begin{proof}
Assuming a syntactic $t : \Tm_i\,\emptycon\,\bot_i$, we can interpret it in the
previously given model, which yields an inhabitant of the metatheoretic $\bot$,
hence a contradiction.
\end{proof}

\section{First-Class Universe Levels}
\label{sec:ttfl}

In the following, we specify and model type theories where levels and
their morphisms are represented by internal types.

However, it would be awkward to pick a particular structure for levels, and
specify a type theory which internalizes that structure; for example
internalizing levels as natural numbers. We do not want to repeat the
specification and semantics for each choice of level structure; instead, we aim
to have a more generic solution.
\begin{enumerate}
  \item We first give a specification of \emph{type theory with dependent
        levels}, or TTDL, where levels and level morphisms may depend on typing contexts.
        Here, liftings, universes and type formers are specified, but the internal
        structure of levels is not yet pinned down.
  \item We show that we can extend TTDL with \emph{level reflection} rules, which identify
        levels with particular internal types, thereby getting \emph{type theories with first-class
        levels}, or TTFL.
\end{enumerate}
This decreases the amount of work that we have to do, in order to get semantics
for different level setups. We only need to pick an external level structure
such that it can be also represented using TTDL type formers.

\begin{mydefinition} A model of TTDL consists of the following.
\begin{enumerate}
\item
  A base category ($\Con$, $\Sub$) with terminal object $\emptycon$.
\item
  A ``dependent'' level structure on the base category:
  \begin{alignat*}{3}
    &\Lvl                  &&: \Con \to \Seti\\
    &\blank\!<\!\blank     &&: \{\Gamma : \Con\} \to \Lvl\,\Gamma \to \Lvl\,\Gamma \to \Seti\\
    &<\!\msf{prop}         &&: (p\,q : i < j) \to p = q\\
    &\blank\!\circ\!\blank &&: j < k \to i < j \to i < k
  \end{alignat*}
  Additionally, $\Lvl$ and $\blank\!<\!\blank$ are natural in the base category,
  so they support substitution operations. \emph{Remark:} at this point, we do not require well-foundedness
  for $\blank\!<\!\blank$, as it has no bearing on basic lifting and universe rules, and
  well-foundedness will be usually internally provable when we add level reflection rules.
\item
  A ``bootstrapping'' assumption on levels. This can be any non-empty collection
  of levels and morphisms. It will be used shortly in Section
  \ref{sec:level_reflection}, where we specify first-class levels using the syntax
  (i.e.\ the initial model) of TTDL. Without bootstrapping, the syntax is
  trivial and has no closed types. Of course, models of TTDL in general make
  sense without the bootstrapping assumption.

  We pick the assumption that $l_0, l_1 : \Lvl\,\Gamma$ exist together with
  $l_{01} : l_0 < l_1$. This allows large eliminations on type formers, so it
  provides a fair amount of power for specifying internal levels.
\item
  A family structure:
  \begin{alignat*}{3}
    &\Ty &&: (\Gamma : \Con) \to \Lvl\,\Gamma \to \Seti\\
    &\Tm &&: (\Gamma : \Con)\{i : \Lvl\,\Gamma\} \to \Ty\,\Gamma\,i \to \Seti\\
    &\blank\ext\blank &&: (\Gamma : \Con)\{i : \Lvl\,\Gamma\} \to \Ty\,\Gamma\,i \to \Con
  \end{alignat*}
  We have type and term substitution, which depends on level substitution. For instance, we have:
  \[
    \blank[\blank] : \Ty\,\Delta\,i \to (\sigma : \Sub\,\Gamma\,\Delta) \to \Ty\,\Gamma\,(i[\sigma])
  \]
  We also have a comprehension isomorphism $\Sub\,\Gamma\,(\Delta \ext A) \simeq
  (\sigma : \Sub\,\Gamma\,\Delta)\times \Tm\,\Gamma\,(A[\sigma])$, which is
  natural in $\Gamma$.
\item A lifting structure with
  \begin{alignat*}{3}
    &\Lift &&: \{\Gamma : \Con\}\{i\,j : \Lvl\,\Gamma\} \to i < j \to \Ty\,\Gamma\,i \to \Ty\,\Gamma\,j\\
    &\up   &&: \{\Gamma : \Con\}\{i\,j : \Lvl\,\Gamma\}(p : i < j) \to \Tm\,\Gamma\,A \to \Tm\,\Gamma\,(\Lift\,p\,A)
  \end{alignat*}
  Such that
  \begin{enumerate}
    \item $\Lift$ preserves all basic type formers and has functorial action on $p \circ q$.
    \item $\up$ has an inverse $\down$, preserves all basic term formers and has functorial action on $p \circ q$.
    \item $(\Gamma \ext A) = (\Gamma \ext \Lift\,p\,A)$, and $\Tm\,\Gamma\,A = \Tm\,\Gamma\,(\Lift\,p\,A)$ and $\up\,t = t$.
  \end{enumerate}

  Above we mention basic type formers, although we have not yet specified
  those. The way this should be understood, is that any basic type former
  introduced from now on should come equipped with preservation equations for
  lifting. This is similar to how we mandate that any introduced type former
  must be natural with respect to substitution.
\item A universe structure
  \begin{alignat*}{6}
    &\U  &&: \{\Gamma : \Con\}(i\,j : \Lvl\,\Gamma) \to i < j \to \Ty\,\Gamma\,j\hspace{2em}
    &\El\,\,&&: \Tm\,\Gamma\,(\U\,i\,j\,p) \to \Ty\,\Gamma\,i
  \end{alignat*}
  such that $\Lift\,p\,(\U\,i\,j\,q) = \U\,i\,k\,(p \circ q)$, $\El$ has
  inverse $\Code$, $\Tm\,\Gamma\,(\U\,i\,j\,p) = \Ty\,\Gamma\,i$ and $\El\,t =
  t$.
\item Basic type formers.
\end{enumerate}
\end{mydefinition}

\begin{mydefinition}[Inductive-recursive model of TTDL]\label{def:ttdlmodel}
 Assume an external $\Lvl$ structure that supports $l_0, l_1 : \Lvl$ and $l_{01}
 : l_0 < l_1$ (the bootstrapping assumption). We again use the universe
 constructions from Section \ref{sec:inductive_recursive_codes}, instantiated to
 the assumed $\Lvl$ structure. We describe components of the model in
 order. Again, we write components of the model in \textbf{bold} font.
\begin{enumerate}
\item The base category remains unchanged from the TTGU model.
\item For the level structure, we define $\bm{\Lvl}\,\Gamma = \Gamma \to \Lvl$ and
  $i \bm{<} j = (\gamma : \Gamma) \to i\,\gamma < j\,\gamma$. Subsitution for internal
  levels and morphisms is given by function composition with $\sigma : \Gamma \to \Delta$.
  Internal composition and $\bm{<\!\msf{prop}}$ follow from the external counterparts.
\item The internal bootstrapping assumption is modeled with the external counterpart.
\item We define $\bm{\Ty}\,\Gamma\,i = (\gamma : \Gamma) \to \U\,(i\,\gamma)$ and
  $\bm{\Tm}\,\Gamma\,A = (\gamma : \Gamma) \to \El\,(A\,\gamma)$. Substitution is
  again function composition, and we have $\Gamma \bm{\ext} A = (\gamma :
  \Gamma)\times \El\,(A\,\gamma)$.
\item Type lifting is given by $\bm{\Lift}\,p\,A = \lambda\,\gamma.\,
  \Lift\,(p\,\gamma)\,(A\,\gamma)$.  Similarly as in the TTGU model,
  $\bm{\Tm}\,\Gamma\,A = \bm{\Tm}\,\Gamma\,(\bm{\Lift}\,p\,A)$ and $\Gamma \bm{\ext}
  A = \Gamma \bm{\ext} (\bm{\Lift}\,p\,A)$ follow from the $\msf{ElLift}$
  equality, and term lifting is the identity function.
\item We define $\bm{\U}\,i\,j\,p = \lambda\,\gamma.\,\U'\,(i\,\gamma)\,(p\,\gamma)$. Again,
  we have $\bm{\Tm}\,\Gamma\,(\U\,i\,j\,p) = \bm{\Ty}\,\Gamma\,i$ by Lemma \ref{lem:ucomp},
  and $\bm{\El}$ and $\bm{\Code}$ are identity functions.
\item Basic type formers are interpreted using $\uir$ codes. Preservation of
  type and term formers by lifting follows by the definition of $\Lift$ and
  $\El$.
\end{enumerate}
\end{mydefinition}

To summarize, the only interesting change compared to the TTGU model is that
levels and level morphisms gain potential dependency on contexts. However, in
the inductive-recursive model this is simply the addition of an extra semantic
function parameter.

\subsection{Level Reflection}\label{sec:level_reflection}

\begin{mydefinition}[Level reflection rules]
Assume that we have definitions for internal levels in the syntax of TTDL,
i.e.\ all of the following are defined:
\begin{alignat*}{3}
  &\Lvl^I               &&: \Ty\,\Gamma\,l_0\\
  &l_0^I,\,l_1^I         &&: \Tm\,\Gamma\,\Lvl^I \\
  &\blank\!<^I\!\blank  &&: \Tm\,\Gamma\,\Lvl^I \to \Tm\,\Gamma\,\Lvl^I \to \Ty\,\Gamma\,l_0\\
  &l_{01}^{I}            &&: \Tm\,\Gamma\,(l_0^I <^I l_1^I)
\end{alignat*}
A \emph{reflection rule} for the above consists of
\begin{enumerate}
  \item $\mkLvl : \Tm\,\Gamma\,\Lvl^I \to \Lvl\,\Gamma$ with its inverse $\unLvl$, such that $\mkLvl\,l_0^I = l_0$ and $\mkLvl\,l_1^I = l_1$.
  \item $\mkMor\,\,: \Tm\,\Gamma\,(i <^I j) \to \mkLvl\,i < \mkLvl\,i$ with its inverse
        $\unMor$.
\end{enumerate}
\end{mydefinition}

For any definition of internal levels, we may extend the specification of TTDL
with the corresponding reflection rule, thereby getting an algebraic signature
for a type theory with first-class levels (TTFL). We can easily get a TTFL with
an inductive-recursive model in the following way. First, we pick an external
$\Lvl$ structure which a) satisfies the bootstrapping assumption b) has sets of
levels and morphisms which can be represented with syntactic TTDL types.

For example, if $\Lvl$ = ($\Nat$, $\blank\!<\!\blank$), with $l_0 = 0$ and $l_1
= 1$, and TTDL supports natural numbers, then we can define $\Lvl^I$ as the
internal $\Nat_{l_0}$, and define $\blank\!<^I\!\blank$ as the usual ordering of
numbers, using TTDL type formers and large elimination (which is available from
$l_0 < l_1$). Then it follows that the model in Definition \ref{def:ttdlmodel},
instantiated to the current level structure, satisfies level reflection. The
model even supports the stricter $\Tm\,\Gamma\,\Nat_{l_0} = \Lvl\,\Gamma$
equation, but in general it is easier to set up models if only an isomorphism
is required.

\subsection{Universe Features in TTFL}

We describe some of the features expressible in TTFL.

\textbf{Bounded universe polymorphism} is realized by quantifying over levels
and morphisms with the usual $\Pi$ types. For example, if levels strictly
correspond to internal natural numbers, we may have
\begin{alignat*}{3}
  &\msf{idUpTo3} : \Pi(l : \Nat_3)(p : \Lift_0^3(l <^I 3))(A : \U\,l\,3\,(\mkMor\,p)) \to \Lift\,(\mkMor\,p)\,A \to \Lift\,(\mkMor\,p)\,A\\
  &\msf{idUpTo3} = \lambda\,l\,p\,A\,a.\,a
\end{alignat*}
Here, we make sure that all types are on the same level, by appropriate
lifting. We assume that internal levels are in $\Nat_0$, but we can bind an $l :
\Nat_3$, because by cumulativity $l$ is also a term of $\Nat_0$. Likewise, the
$p$ variable is a term of $\Lift_0^3(l <^I 3)$ and $l <^I 3$ as well.

\textbf{Transfinite hierarchies} are naturally supported. For example, $\Lvl$
can be identified with $\msf{Maybe}\,\Nat_0$, where $\msf{Nothing}$ defines
$\omega$ and $\msf{Just}\,n$ is a finite level. Then, by the definition of
morphisms, we have $<\!\omega : \Pi(n : \Nat_0) \to \msf{Just}\,n <^I \omega$.
We can use this to quantify over finite levels, as in the following type:
\begin{alignat*}{3}
  \Pi(n : \Nat_{\omega})(A : \U\,n\,\omega\,(\mkMor\,(<\!\omega\,n))) \to \Lift\,(\mkMor\,(<\!\omega\,n))\,A \to \Lift\,(\mkMor\,(<\!\omega\,n))\,A
\end{alignat*}
This type is in $\Ty\,\Gamma\,\omega$, but it is not in any universe, since
$\omega$ is the greatest level.

\textbf{Induction on levels and level morphisms.} In Agda 2.6.1, there is an
internal type of finite levels, and while construction rules and some built-in
operations on levels are exposed, there is no general elimination rule on
levels. Thus, there is a $\Nat \to \Lvl$ conversion function but it has no
inverse. In contrast, TTFL supports arbitrary elimination on levels and
morphisms.

\textbf{Type formers returning in least upper bounds of levels}. It is common in
type theories to allow type formers to have parameter types in different
universe levels, say $i$ and $j$, and return in level $i \sqcup j$. In TTFL,
whenever levels are \emph{trichotomous}, meaning that the ordering and equality
of levels is internally decidable, $i \sqcup j$ can be defined as the greater of
$i$ and $j$, and the ``heterogeneous'' type formers are derivable\footnote{A
level structure which is trichotonomous and supports \emph{extensionality},
i.e.\ $(\forall i.\, (i < j) \iff (i < k)) \to j = k$, is a
\emph{type-theoretic ordinal}. Assuming excluded middle, type-theoretic
ordinals are equivalent to classical ordinals \cite[Section 10.3]{hottbook}.}.

\textbf{Coercive cumulative subtyping.}
TTFL as specified does not directly support cumulative subtyping. However, it is
compatible with coercive subtyping. Consider the following rules:
\begin{alignat*}{3}
  &\blank\!\leq\!\blank &&: \Ty\,\Gamma\,i \to \Ty\,\Gamma\,j \to \Seti\\
  &\msf{coerce} &&: A \leq B \to \Tm\,\Gamma\,A \to \Tm\,\Gamma\,B\\
  &\leq\!\msf{refl} &&: A \leq A \\
  &\U\!\leq    &&: i < i' \to \U\,i\,j\,p \leq \U\,i'\,k\,q\\
  &\Pi\!\leq   &&: (p : A' \leq A) \to ((a' : \Tm\,\Gamma\,A')
  \to B[x \mapsto \msf{coerce}\,p\,a'] \leq B'[x \mapsto a'])\\
  & && \hspace{1em} \to \Pi(x : A) B \leq \Pi (x : A') B'
\end{alignat*}
Any model of TTFL can support the above rules: we can define
$\blank\!\leq\!\blank$ and $\msf{coerce}$ by indexed induction-recursion
\cite{indexedir}, where we define coercion along $\U\!\leq$ by type lifting, and
coercion along $\Pi\!\leq$ by backwards-forwards coercion. It is possible to
extend the subtyping relation with rules for other basic type formers.

Note that $\Pi$ is contravariant in the domain. This is easily supported with
our inductive-recursive semantics, unlike in the set-theoretic model of
cumulativity for Coq \cite{timany18cumulative}, where function domains are
invariant.

\subsection{Effects of Choice of Level Structure}

TTFL features clearly vary depending on level structures. We make some basic
observations.

\begin{itemize}
\item We did not mandate that the level of $\Lvl^I$ is the least level,
i.e.\ that $l_0 < i$ for every $i \neq l_0$. If this holds, then it is
possible to have level polymorphism at every level: at $l_0$ we can just bind
a $\Lvl^I$, and at every other level, we can lift $\Lvl^I$ to that
level. However, levels are not necessarily totally ordered, and $l_0$ does not
have to be the least. This means that universe polymorphism is prohibited in
levels which are not connected to $l_0$.

\item If levels are given by a limit ordinal, then every TTFL type is contained
in a universe. If levels form a successor ordinal, then this is not the
case. For example, Agda 2.6.1 has $\omega + 1$ levels (externally), where
$\Set{\omega}$ is the topmost universe, but $\Set{\omega}$ is not in any
universe.

\item
While it is possible to quantify over all levels (using plain $\Pi$ types), it
is not possible to have level polymorphism over all levels. We may try to type
an identity function for all levels, as $\Pi(i : \Lift\,?\,\Lvl^I)(A :
\U\,(\mkLvl\,i)\,?\,?) \to \Lift\,?\,A \to \Lift\,?\,A$. The issue is in
$\U\,(\mkLvl\,i)\,?\,?$, where we would have to find a level which is larger
than \emph{every} level. The solution to this issue is to simply add more
levels. For example, for polymorphism over finite levels, we may pick $\omega +
\omega$ as the first limit ordinal which can internalize finite level
polymorphism; this is what Agda 2.6.2 does.
\end{itemize}

\section{Related Work} \label{sec:related}

Predicative hierarchies originate from Russell's ramified type theories
\cite{principia}. In the more modern formulations of type theory, Martin-Löf
proposed a countable predicative hierarchy \cite{martinlof73predicative}, as a
way to remedy the inconsistency of the previous version of the theory (which
assumed type-in-type). Harper and Pollack described universe inference with
level assignments and also a form of level polymorphism \cite{harperpollack}.
Sterling \cite{sterling2019algebraic} gave an algebraic specification much like
ours for a type theory with countable cumulative universes, and proved canonicity
for it.

There have been proposals for strengthening universes with various closure
principles and universe operators. Palmgren's super universes and higher-order
universes \cite{Palmgren98onuniverses} and Setzer's Mahlo universes
\cite{setzer00mahlo} are examples for this. These are sufficient to model
transfinite hierarchies, but as we noted in Section
\ref{sec:inductive_recursive_codes}, we do not know how to model strict
inclusions with them. Variants of induction-recursion \cite{dybjer99finite,
  indexedir, positiveir} are particularly flexible and powerful extensions to
universes. McBride gave an inductive-recursion definition of cumulative
universes that we adapted in this work \cite{mcbride2015datatypes}.

It is worth to summarize here the universe features in the current
type theory implementations.

\textbf{Agda 2.6.1} has $\omega+1$-many non-cumulative predicative universes as
$\Set{i}$, with optional cumulative subtyping only for universes
\cite{agdadocs}.  It also has an internal type $\msf{Level} : \Set0$ for finite
levels (hence, exluding $\omega$), which supports constructors and some built-in
operations, but no general elimination rule. There is also a countable parallel
hierarchy $\msf{Prop_i}$ for strict propositions \cite{sprop}. Agda 2.6.2 will
extend the $\Set{i}$ hierarchy to $\omega * 2$.

\textbf{Coq 8.13} has $\omega$-many cumulative predicative universes with
cumulative subtyping for all type formers \cite{timany18cumulative}. It supports
bounded universe polymorphism, but it has no internal type for levels, and
universe polymorphic definitions are not internally typeable. It also has an
impredicative $\msf{Prop}$ universe and optionally impredicative bottom $\Seti$
universe. Version 8.13 added experimental support for a parallel
countable cumulative hierarchy for strict propositions.

\textbf{Lean 3.3} has countable non-cumulative predicative $\msf{Type}_i$
universes with universe polymorphism, and no internal type of levels
\cite{leanmanual}. It also has strict impredicative $\msf{Prop}$.

\textbf{Idris 1} has countable cumulative predicative universes with cumulative
subtyping only for universes, typical-ambiguity-style level inference and no
universe polymorphism \cite{idrisdocs}.

Of the above features, what TTFL does not support is a) impredicativity b) the
interaction of $\msf{Prop}$ and $\msf{Type}$ universes, i.e. the restrictions on
$\msf{Prop}$ elimination.

\section{Conclusion and Future Work} \label{sec:conclusion}

In the current work, we developed a framework for modeling a variety of universe
features in type theories. At this point, we may ask the question: if
induction-recursion is sufficient to model every feature, why not simply support
it in a practical implementation, and drop the menagerie of universe features?

The answer is that induction-recursion provides a \emph{deep embedding} of
universe features, which is usually less convenient to use than \emph{native}
features. For example, both Coq and Agda have powerful automatic solving for
filling out implicit universe levels. We also do not have to invoke $\El$ or the
$\ult$ computation rule explicitly, and in Coq we can use implicit syntax for
subtyping instead of explicit coercions.

This trade-off between convenience and formal minimalism is similar to the
situation with inductive types. Formally, W-types and identity types are easier
to handle than general inductive families, but the latter are far more
convenient to actually use. Ideally, we would like to justify complicated
convenience features by reduction to minimal features. With the current paper,
we hope to have made progress in this manner.

\subsection{Future Work}

Several related topics are not discussed in this paper and could be subject to
future work.

First, besides consistency, we are often interested in \emph{canonicity},
\emph{normalization} or other metatheoretical properties. The current work
focuses on consistency and leaves other properties to future work. We did keep
canonicity in mind when specifying the systems in this paper. Hopefully the
usual proof method of gluing (in other words, proof-relevant logical predicates)
\cite{coquandnorm,kaposi2019gluing,sterling2019algebraic} can be adapted to the theories in
this paper.

Second, we only focus on using universes as size-based classifiers for
types. Stratification features are also present in two-level type theory
\cite{twolevel}, modal type theories \cite{gratzer20multimodal} or as h-levels
in homotopy type theory \cite{hottbook}. It would be interesting to port
universe features in this paper to two-level type theory, as they would
hopefully model a form of stage polymorphism in multi-stage compilation. We
could try representing $\msf{Prop}$ universes in TTFL as well. This is closely
related to h-level based stratification.

Third, we do not discuss implementation strategies and ergonomics of universe
features. Which universe hierarchies support good proof automation? What kind of
impact do first-class levels have on elaboration algorithms? Hopefully the
current work can aid answering these questions, by at least giving a way to quickly
check if some features are logically consistent.

Lastly, we do not handle impredicative universes. The main reason for this is
that we do not know the consistency of having induction-recursion and
impredicative function space together in the same universe, and modeling
impredicativity seems to require this assumption in the metatheory. This could
be investigated as well in future work.

\bibliography{references.bib}
\end{document}